\documentclass{article}
\usepackage{graphicx} 

\usepackage[utf8]{inputenc}
\usepackage[T1]{fontenc}
\usepackage[english]{babel}
\usepackage{amsthm}
\usepackage{amsmath}
\usepackage{enumitem}
\usepackage{amssymb}
\usepackage{enumitem}
\usepackage{graphicx}
\usepackage{caption}    
\usepackage{subcaption} 
\usepackage{multirow}
\usepackage{makecell}
\usepackage{braket}
\usepackage{placeins}
\usepackage[hidelinks]{hyperref}

\newtheorem{theorem}{Theorem}[section]
\newtheorem{lemma}{Lemma}[section]
\newtheorem{proposition}{Proposition}[section]

\DeclareMathOperator\supp{supp}

\title{New Methods for Constructing Classical and Quantum Codes from Graphs and Matroids}
\author{Anderson S. Barbosa$^{\dagger}$ \and Franklin de L. Marquezino$^{\ddagger}$ \and Giuliano G. La Guardia$^{\star}$\\[2em]
$^{\dagger\ddagger}$Universidade Federal do Rio de Janeiro, Brazil\\
$^{\star}$Universidade Estadual de Ponta Grossa, Brazil\\[2em]
anderson@cos.ufrj.br, franklin@cos.ufrj.br, gguardia@uepg.br}
\date{}

\begin{document}

\maketitle

\begin{abstract}
    In this paper, we present a method for constructing linear codes derived from book graphs. Applying  such a method, we construct binary codes that have the greatest possible minimum distance for their length and dimension by means of the $2$-sum of two equal book graphs. We then utilize matroid and graph theory together to investigate some properties of the resulting linear codes. Furthermore, we show that, from the incidence matrix of any graph, we can obtain a linear code over any finite field maintaining the same parameters. Furthermore, we show how the condition for two linear codes to generate a \textit{Calderbank-Shor-Steane} (CSS) code translates into matroid theory and, utilizing cycle graphs, we exhibit a way to generate codes that meet this condition.
\end{abstract}

\section{Introduction}

Classical error-correcting codes (ECCs)~\cite{Huffman_Pless_2003} are essential in today's digital communication, because they allow information to be sent reliably through noisy channels. As it is well-known, constructions of (classical) linear codes are much investigated in the literature, due to the structure of vector space they have. We denote by $[n, k, d]_q$ the parameters of a linear code over the field ${\mathbb F}_q$, where $q$ is a prime power, $n$ is the code length, $k$ is the code dimension and $d$ is the minimum Hamming distance of the code.  One of the main goals of research in this area is to find codes with higher dimension and minimum distance increasing the length as little as possible.

The study of linear codes is also important in the development of quantum error-correcting codes (QECCs), especially the class of 
\textit{Calderbank-Shor-Steane} (CSS) codes~\cite{calderbank_shor_css,steane_css}, since two compatible classical codes can be used to construct a CSS code. Similarly to classical codes, the parameters of a quantum CSS code is denoted by $[[n, k, d]]_q$, where $n$ is the length, $k$ is the number of logical qudits and $d$ is the minimum distance of the code.
Here, the aim is the same: to increase the code dimension and the minimum distance for a fixed $n$. Then, as the properties of a CSS code depend on the properties of the linear codes used to construct them, research on classical codes also impact the quantum counterpart.

Code theory shares a close relation with graph theory~\cite{code_graph_1,code_graph_2} and matroid theory~\cite{Whitney1987,greene}, hence several works on linear codes that take advantage of these mathematical structures have emerged. For example, we have graph-adjacency codes~\cite{hoffman_higman_graphs,TONCHEV2002549}, that have adjacency matrices as their generator matrix and low density parity-check codes based on matroids~\cite{ldpc_matroids}. 

In this paper, we investigate linear codes obtained as the null space of the incidence matrix of graphs and their relation with graphic matroids. Previous works have investigated linear codes from incidence matrices, but focused on specific classes of graphs, such as $k$-regular graphs~\cite{codes_from_incidence,codes_from_incidence_2} and unit graphs~\cite{codes_from_unit,codes_from_unit_2}. Using this approach, we present a method for constructing linear codes from the 2-sum of book graphs. Using such construction, we were able to find four binary codes that have the greatest minimum distance for their length and dimension. They are $[8,3,4]_2$, $[12,5,4]_2$, $[12,3,6]_2$, $[16,3,8]_2$ linear codes.

A $\mathbb{F}$-representable matroid can be related to a linear code over $\mathbb{F}$ through its parity-check matrix or its generator matrix. For quantum codes, on the other hand, the closest thing we currently know is that there is a correspondence between them and symplectic matroids~\cite{quantum_symplectic}. In order to establish a connection between quantum codes, specifically CSS codes, and ordinary matroids, we proved a sufficient condition for the codes corresponding to two matroids to be used to construct a CSS code. With this, we can generate linear codes through matroids and have a way to guarantee that codes generated can generate a CSS code as well. Furthermore, we showed how cycle graphs can be used to achieve this condition, yielding CSS codes with distance 2.

Finally, we show that a code constructed as the null space of the incidence matrix of orientation of a graph can be easily generalized to any finite field. This opens up a wide range of possibilities, since we can have any graph structure generate a linear code over a field with any characteristic we want. For example, the online database Code Tables~\cite{Grassl:codetables} created by Markus Grassl, which is one of the most widely recognized references for the best known linear codes, only registers codes over field $\mathbb{F}_q$ for $2\leq q \leq 9$, and our generalization allows us to go way beyond that.

The paper is organized as follows. Section~\ref{sec:background} brings a brief review of graphs, linear codes and matroid theory, that are the bases of this work. In Section~\ref{sec:results}, we present the new results of the paper. In particular, we construct new linear codes with maximum minimum distance (see Theorem~\ref{theorem:new_codes}). Finally, in Section~\ref{fr}, the final remarks are drawn.

\section{Preliminaries}
\label{sec:background}

\subsection{Graph Theory}

An \textit{(undirected) graph} $G$ is an ordered pair $(V, E)$, where $V$ is a set of objects called \textit{vertices} and $E$ is a collection of unordered pairs of vertices, called \textit{edges}. A vertex $v$ is incident with an edge $e$ if $v$ belongs to $e$. A \textit{directed graph} $D$ is an ordered pair $(V, A)$, where $V$ is a set of objects called \textit{vertices} and $E$ is a collection of ordered pairs of vertices, called \textit{arcs}. If $(x,y)$ is an arc, we call $x$ the \textit{tail} and $y$ the \textit{head} of the arc. An \textit{orientation} of a graph $G$ is a directed graph $D$ obtained by assigning a direction to all edges of $G$. Two graphs $G$ and $H$ are \textit{isomorphic} if there are bijections $\psi : V(G) \to V(H)$ and $\theta : E(G) \to E(H)$ such that $v$ is incident with $e$ if and only if $\psi(v)$ is incident with $\theta(e)$~\cite{oxley}.

One way to represent a graph is through an \textit{adjacency matrix}, which is a $|V|\times|V|$ matrix whose rows and columns represent vertices of the graph. Let $G=(V,E)$ be a graph. Then, its adjacency matrix $A$ is built as follows:
\begin{equation}
    a_{ij}=\begin{cases}
        1, & \text{if }(i,j) \in E\\
        0, & \text{otherwise}
    \end{cases}.
\end{equation}
Notice that the adjacency matrix of a directed graph is constructed almost the same way as the one of an undirected graph. The only difference is that $a_{i,j}$ is not necessarily equal to $a_{j,i}$, since $(i,j)\neq(j,i)$. 

An alternative way of representing a graph is through an \textit{incidence matrix}, which is a $|V|\times |E|$ matrix whose rows represent the vertices and columns represent the edges of the graph. For a directed graph $D=(V,A)$, the incidence matrix $A_{I}(D)$ is built as follows.
\begin{equation}
    {a_I}_{ij}=\begin{cases}
        1, & \text{if $i$ is the tail of arc $j$}\\
        -1, & \text{if $i$ is the head of arc $j$}\\
        0, & \text{otherwise}
    \end{cases}.
\end{equation}
In the case of an undirected graph, $a_{Iij}$ is 1 if vertex $i$ is incident with edge $j$ and 0 otherwise.

A \textit{subgraph} of a graph $G$ is a graph $H$ such that $V(H)\subseteq V(G)$ and \mbox{$E(H)\subseteq E(G)$}. If $H \neq G$, we say $H$ is a \textit{proper subgraph} of $G$. Let $G$ be a graph and $V'\subseteq V(G)$ a subset of the set of vertices of $G$. Then the subgraph $G'$ of $G$ \textit{induced} by $V'$ is the graph with set of vertices $V'$ and set of edges formed by the edges of $G$ with both ends in $V'$. We say that $G'$ is an \textit{induced subgraph} of $G$. 

Now, let $E' \subseteq E(G)$ be a subset of the set of edges of a graph $G$. Then the subgraph $G''$ \textit{induced} by $E'$ is a graph with set of edges $E'$ and set of vertices equal to set of ends of the edges in $E'$. We say $G''$ is an \textit{edge-induced subgraph} of $G$.

A \textit{walk} is a finite sequence $W=v_0,e_1,v_1,e_2,v_2,\dots,e_kv_k$ that alternates vertices and edges, starting and ending in a vertex, such that, for $1\leq i \leq k$, $e_i=(v_{i-1},v_i)$. A \textit{trail} is a walk where all edges are distinct. A \textit{path} is a trail where all vertices are distinct. A \textit{circuit} is a trail that starts and ends in the same vertex. A circuit in which all the other vertices are distinct is called a \textit{cycle}. The number of edges of the smallest cycle of a graph $G$ is its \textit{girth} and is denoted by $g(G)$. An \textit{acyclic graph} is a graph that contains no cycle.

A graph $G$ is \textit{connected} if, for all $v_i, v_j \in V(G)$, with $i\neq j$, there exists a path between $v_i$ and $v_j$ and is \textit{disconnected} otherwise. A \textit{connected component} of a graph $G$ is a maximal connected subgraph of $G$. Clearly the sets of vertices of the connected components of $G$ form a partition of $V(G)$~\cite{diestel2025graph}. The number of connected components of a graph $G$ is denoted by $\omega(G)$.

A \textit{book graph} $B_{n,c}$ is a graph such that there exist exactly $c$ cycles of length $n$ sharing only a single edge~\cite{book_graph}. Figure~\ref{fig:book-graph} shows a graphical representation of $B_{4,2}$. Notice that it has two four-edge cycles, $\{(1,2),(2,4),(4,3),(3,1)\}$ and $\{(1,2),(2,6),(6,5),(5,1)\}$, sharing the edge $(1,2)$. Proposition~\ref{prop:book-graph-props} in Section~\ref{sec:results} proves some properties of this class of graphs.

\begin{figure}
    \centering
    \includegraphics[width=0.75\linewidth]{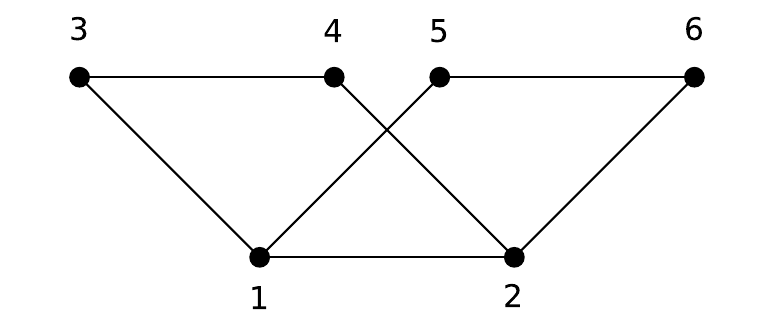}
    \caption{Graphical representation of $B_{4,2}$}
    \label{fig:book-graph}
\end{figure}

\subsection{Matroid Theory}

A matroid $M$ is an ordered pair $(E,\mathcal{I})$, where $E$ is a finite set, called \textit{ground set}, and $\mathcal{I}$ is a collection of subsets of $E$ with the following properties:
\begin{enumerate}[label=\textbf{(I\arabic*)}]
    \item $\emptyset \in \mathcal{I}$.
    \item If $I_1 \in \mathcal{I}$ and $I_2 \subseteq I_1$ and , then $I_2 \in \mathcal{I}$.
    \item If $I_1, I_2 \in \mathcal{I}$ and $|I_1| < |I_2|$, then there exists an element $x\in I_2-I_1$ such that $I_1 \cup \{x\} \in \mathcal{I}$.
\end{enumerate}

The subsets of $E$ which are in $\mathcal{I}$ are said to be \textit{independent sets} of $M$ and the subsets of $E$ that are not in $\mathcal{I}$ are called \textit{dependent sets} of $M$. Two matroids $M_1=(E_1,\mathcal{I}_1)$ and $M_2=(E_2,\mathcal{I}_2)$ are \textit{isomorphic} if there is a bijection $\psi: E_1 \to E_2$ such that, for all subsets $X\subseteq E_1$, the set $\psi(X)$ is a member of $\mathcal{I}_2$ if and only if $X$ is a member of $\mathcal{I}_1$. From the definition of a matroid, it is straightforward to construct a matroid from a matrix. The following proposition show how this can be done~\cite{oxley}.

\begin{proposition}
    \label{prop:vector-matroid}
    \textup{\textbf{(Oxley~\cite[Page 8]{oxley})}} Let $A$ be a $m\times n$ matrix over a field $\mathbb{F}$ and $E$ be the set of column labels of $A$. Let $\mathcal{I}$ be the collection of subsets $I$ of $E$ for which the multiset of columns labelled by $I$ is linearly independent in the vector space $V(m,\mathbb{F})$. Then $(E,\mathcal{I})$ is a matroid.
\end{proposition}

The matroid described in Proposition~\ref{prop:vector-matroid} is the \textit{vector matroid} of the matrix $A$, denoted by $M[A]$. A $\mathbb{F}$\textit{-representable} matroid $M$ is a matroid isomorphic to the vector matroid of a matrix over the field $\mathbb{F}$ and we say such a matrix \textit{represents} $M$ \textit{over} $\mathbb{F}$. A minimal dependent set of a matroid $M$ is called a \textit{circuit} of $M$. The set of circuits of $M$ is denoted by $\mathcal{C}$ or $\mathcal{C}(M)$.

\begin{lemma}
\label{lemma:circuit_property}
    \textup{\textbf{(Oxley~\cite[Page 8]{oxley})}} The set of circuits $\mathcal{C}$ of a matroid has the following properties:
    \begin{enumerate}[label=\textbf{(C\arabic*)}]
        \item $\emptyset \notin \mathcal{C}$.
        \item If $C_1, C_2 \in \mathcal{C}$ and $C_1 \subseteq C_2$, then $C_1=C_2$.
        \item If $C_1$ and $C_2$ are distinct members of $\mathcal{C}$ and $e \subseteq C_1 \cap C_2$, then there exists a member $C_3$ of $\mathcal{C}$ such that $C_3 \in (C_1 \cup C_2) \setminus\{e\}$.
    \end{enumerate}
\end{lemma}

The cardinality of the smallest circuit of a matroid $M$ known as its \textit{girth} and denoted by $g(M)$. There are several equivalent ways of representing a matroid and its set of circuits is one of them. The following proposition complements Lemma~\ref{lemma:circuit_property} regarding the characterization of the set of circuits of a matroid.

\begin{proposition}
    \textup{\textbf{(Oxley~\cite[Page 9]{oxley})}} Let $E$ be a set and $\mathcal{C}$ a collection of subsets of $E$ satisfying the properties \textbf{(C1)}, \textbf{(C2)} and \textbf{(C3)} from Lemma~\ref{lemma:circuit_property}. Let $\mathcal{I}$ be the collection of subsets of $E$ that contain no member of $\mathcal{C}$. Then $(E,\mathcal{I})$ is a matroid with $\mathcal{C}$ as its collection of circuits.
\end{proposition}

With this alternative way of representing a matroid, we can construct a matroid from a graph. The proposition below shows how we can do this.

\begin{proposition}
    \label{prop:cycle-matroid}
    \textup{\textbf{(Oxley~\cite[Page 10]{oxley})}} Let $G$ be a graph, $E$ the edge set of $G$ and $\mathcal{C}$ the collection of edge sets of the cycles of $G$. Then $\mathcal{C}$ is the set of circuits of a matroid on $E$.
\end{proposition}

The matroid described in Proposition~\ref{prop:cycle-matroid} is called \textit{cycle matroid} of the graph $G$ and is denoted by $M(G)$. A \textit{graphic matroid} is matroid that is isomorphic to the cycle matroid of a graph.

\begin{lemma}
\label{prop:rank-cyclic}
    \textup{\textbf{(Oxley~\cite[Page 24]{oxley})}} Let $G$ be a graph. Then the rank of the cycle matroid $M(G)$ is $r_{M(G)}=|V(G)|-\omega(G)$.
\end{lemma}

\begin{proposition}
		\label{prop:rep-graphic}
        \textup{\textbf{(Oxley~\cite[Page 135]{oxley})}} Let $D(G)$ be an arbitrary orientation of a graph $G$ and $\mathbb{F}$ a field. Then an incidence matrix $A_{D}$ of $D(G)$ in which the columns corresponding to loops have only zeroes represents $M(G)$ over $\mathbb{F}$.
\end{proposition}

Let $G_1$ and $G_2$ be graphs and $p_1$ and $p_2$, respectively, edges of $G_1$ and $G_2$. For $i\in\{1,2\}$, assign a direction to $p_i$ and label its tail with $u_i$ and its head with $v_i$. A \textit{parallel connection} of $G_1$ and $G_2$ with respect to the directed edges $p_1$ and $p_2$ are built as follows. First, we delete $p_1$ from $G_1$ and $p_2$ from $G_2$. Next, we merge $u_1$ and $u_2$ into a single vertex $u$ and $v_1$ and $v_2$ into a single vertex $v$. Finally, we add an edge $p$ connecting $u$ and $v$.

Let $\mathcal{C}_P$ be the set of circuits of the cycle matroid of the parallel connection of $G_1$ and $G_2$. Denoting the cycle matroids of $G_1$ and $G_2$ by $M_1$ and $M_2$, respectively, and assuming that neither $p_1$ nor $p_2$ is a loop or a bridge, we have

\begin{align}
		\label{eq:paralell-connection-circuits} 
		\begin{split}
			\mathcal{C}_P=&\mathcal{C}(M_1\setminus p_1)\cup \{(C_1\setminus\{p_1\})\cup \{p\}:p_1\in C_1 \in \mathcal{C}(M_1)\}\cup\phantom{}\\
			&\mathcal{C}(M_2\setminus p_2)\cup\{(C_2\setminus\{p_2\})\cup \{p\}:p_2\in C_2 \in \mathcal{C}(M_2)\}\cup\phantom{}\\
			&\{(C_1\setminus\{p_1\})\cup(C_2\setminus\{p_2\}):p_i\in C_i \in \mathcal{C}(M_i)\text{ for all }i\}.
		\end{split}
\end{align}

Figure~\ref{fig:parallel-connection-g1g2} shows an example of a parallel connection of the graphs $G_1$ and $G_2$ presented, respectively, in Figures~\ref{fig:g1-parallel-connection} and~\ref{fig:g2-parallel-connection}. Now consider the graph from Figure~\ref{fig:2sum-g1g2}. Notice that it can be obtained from the parallel connection of $G_1$ and $G_2$ by just deleting the edges $p_1$ and $p_2$, merging the vertices $u_1$ and $u_2$ into a vertex $u$ and $v_1$ and $v_2$ into a vertex $v$. A graph obtained this way is a $2$\textit{-sum} of $G_1$ and $G_2$.

Equation~\ref{eq:paralell-connection-circuits} shows how to compute the cycle matroid of the parallel connection of two graphs $G_1$ and $G_2$ from $M(G_1)$ and $M(G_2)$, however, we will see that it can also be used to generalize the parallel connection to two arbitrary matroids. Suppose that $M_1$ and $M_2$ are two arbitrary matroids such that $E(M_1)\cap E(M_2)=\emptyset$. Let $p_1 \in E(M_1)$ be an element of $M_1$ and $p_2 \in E(M_2)$ an element of $M_2$ such that neither $p_1$ nor $p_2$ are a loop or a coloop. Let $p$ be an element that are neither in $E(M_1)$ nor in $E(M_2)$ and $E=E(M_1\setminus p_1) \cup E(M_2\setminus p_2) \cup p$. Then, for the collection $\mathcal{C}_P$ specified in Equation~\ref{eq:paralell-connection-circuits}, the proposition below follows.

\begin{proposition}
    \textup{\textbf{(Oxley~\cite[Page 252]{oxley})}} $\mathcal{C}_P$ is the set of circuits of a matroid on $E$.
\end{proposition}

	\begin{figure}[htbp]
		\centering
		
		\begin{subfigure}{0.4\textwidth}
			\centering
			\includegraphics[width=\linewidth]{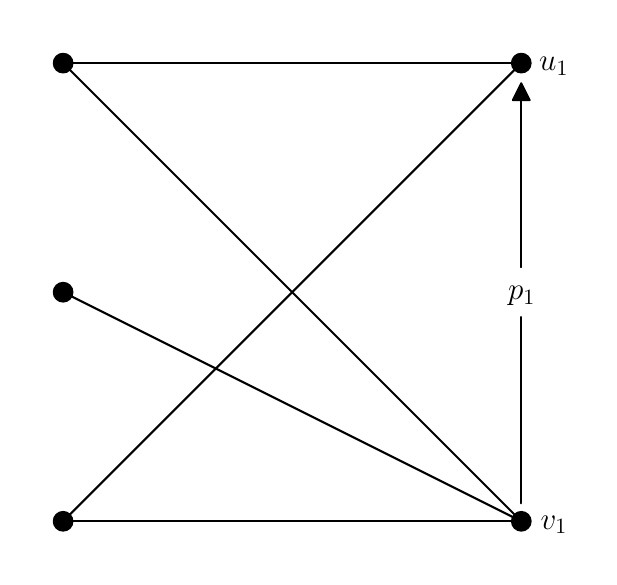}
			\caption{$G_1$}
			\label{fig:g1-parallel-connection}
		\end{subfigure}
		\hfill
		\begin{subfigure}{0.4\textwidth}
			\centering
			\includegraphics[width=\linewidth]{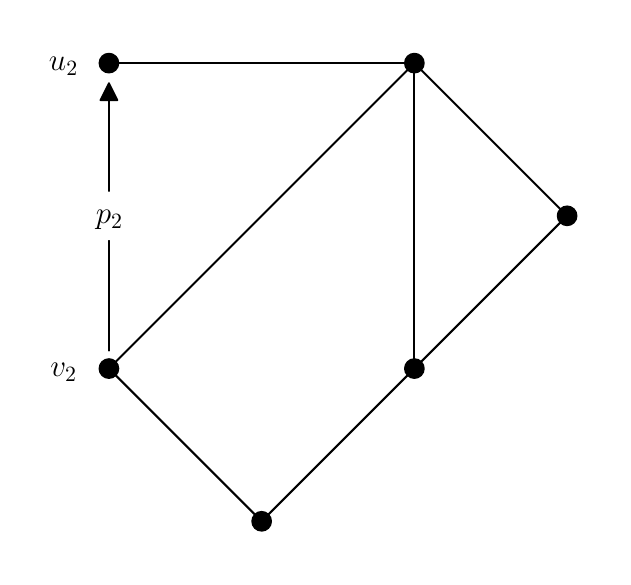}
			\caption{$G_2$}
			\label{fig:g2-parallel-connection}
		\end{subfigure}
		
		\vspace{0.5cm}
		
		\begin{subfigure}{0.48\textwidth}
			\centering
			\includegraphics[width=\linewidth]{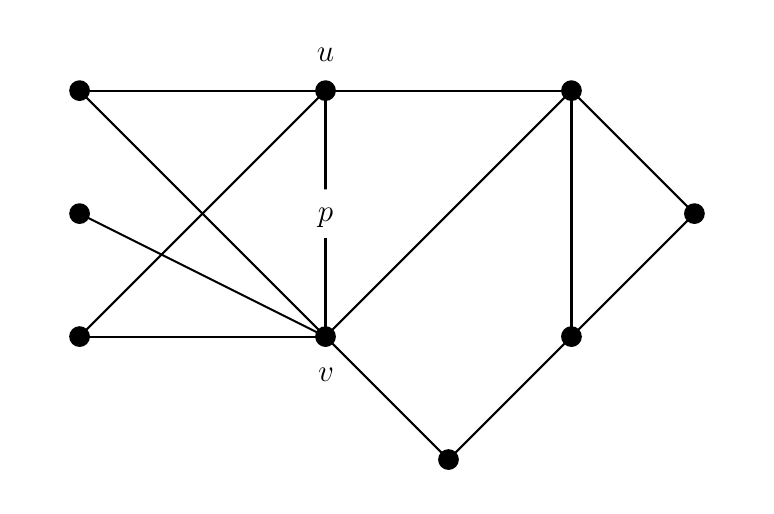}
			\caption{Parallel connection of $G_1$ and $G_2$}
			\label{fig:parallel-connection-g1g2}
		\end{subfigure}
		\hfill
		\begin{subfigure}{0.48\textwidth}
			\centering
			\includegraphics[width=\linewidth]{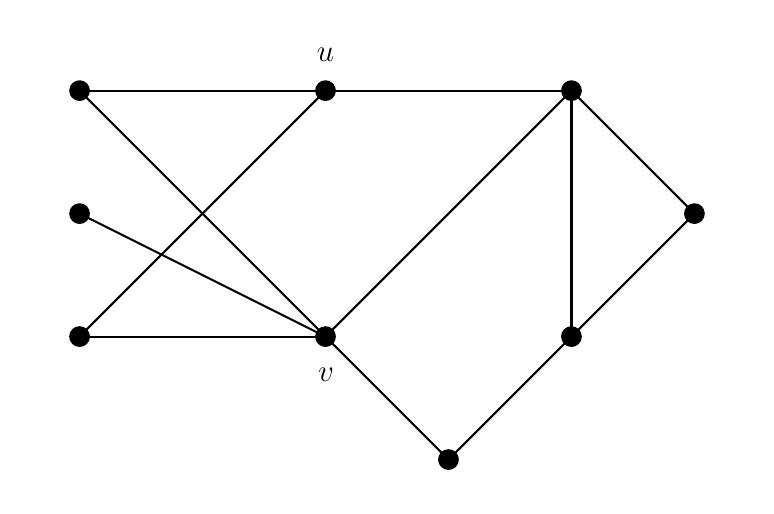}
			\caption{2-sum of $G_1$ and $G_2$}
			\label{fig:2sum-g1g2}
		\end{subfigure}
		
		\caption{2-sum and parallel connection of two graphs}
		\label{fig:2sum-parallel}
	\end{figure}

Let $M$ be a matroid.  Let $\xi$ be a relation on the ground set $E(M)$ such that $e\ \xi\ f$ if $e=f$ or $M$ has a circuits containing $\{e,f\}$. The equivalence classes of $\xi$ are the \textit{(connected) components} of $M$. A \textit{separator} of $M$ is a set $X \subseteq E(M)$ that is the union of components of $M$.

Let $M_1$ and $M_2$ be matroids with at least two elements such that \linebreak\mbox{$E(M_1)\cap E(M_2)=\{p\}$} and that do not have $p$ as a separator. Then the \mbox{$2$-sum} $M_1 \oplus_2M_2$ of $M_1$ and $M_2$ is the matroid with ground set $E=(E(M_1)\cup E(M_2))\setminus\{p\}$ and set of circuits
\begin{equation}
		\mathcal{C}=\mathcal{C}(M_1\setminus p) \cup \mathcal{C}(M_2\setminus p) \cup \{C_1 \cup C_2:p\in C_1 \in \mathcal{C}(M_1),p\in C_2 \in \mathcal{C}(M_2)\}.
\end{equation}

\begin{proposition}
\label{prop:2srank}
    \textup{\textbf{(Jaeger, Vertigan, and Welsh~\cite[Page 38]{Jaeger_Vertigan_Welsh_1990})}} If $M=M_1\oplus_2 M_2$, then the rank of $M$ is $r_M=r_{M_1}+r_{M_2}-1$
\end{proposition}

\subsection{Linear Codes}

A \textit{linear code} $C$ over $\mathbb{F}_q$ is the row space of a matrix $G \in \mathcal{M}_{k\times n}(\mathbb{F}_q)$, where $\mathcal{M}_{k\times n}(\mathbb{F}_q)$ is the space of all $k \times n$ matrices over $\mathbb{F}_q$. We call $G$ the \textit{generator matrix} of $C$. Equivalently, $C$ is the null space of a matrix $H \in \mathcal{M}_{n-k\times n}(\mathbb{F}_q)$, called \textit{parity-check matrix} of $C$. The elements of $C$ are its \textit{codewords}.

The \textit{(Hamming) distance} $\text{d}(x,y)$ between two vectors $x,y \in \mathbb{F}_q^n$ is the number of coordinates where $x$ and $y$ differ. The \textit{(Hamming) weight} $\text{w}(x)$ of a vector $x \in \mathbb{F}_q^n$ is the number of nonzero coordinates of $x$. The \textit{(minimum) distance} of a linear code $C$ is the smallest distance between two codewords of $C$ or, equivalently, the minimum weight of nonzero codewords of $C$. If a linear code $C$ over $\mathbb{F}_q$ has length $n$, dimension $k$ and minimum distance $d$, we denote it by $[n, k, d]_q$ \cite{Huffman_Pless_2003}.

Let $C$ be a $[n, k, d]$ linear code over ${\mathbb F}_q$ with generator matrix $G$ and parity-check matrix $H$. Then the \textit{dual code} $C^\perp$ of $C$ is the $[n,n-k]$ linear code with generator matrix $H$ and parity-check matrix $G$. Equivalently, the dual code $C^\perp$ consists of every vector in ${ \mathbb F}_q^n$ that are orthogonal to all codewords of $C$, that is,
\begin{equation}
		C^\perp=\{x \in \mathbb{F}_{q}^{n}:\langle x,y\rangle=0\text{ for all }y \in C\},
\end{equation}
where $\langle x,y\rangle$ denotes the \textit{Euclidean inner product} of $x$ and $y$, defined by
\begin{equation}
    \langle x,y\rangle = \sum_{i=1}^{n}x_iy_i.
\end{equation}
A linear code $C$ such that $C \subseteq C^\perp$ is called \textit{self-orthogonal}. If $C=C^\perp$, $C$ is called \textit{self-dual}.

\subsection{Stabilizer Codes}

In order to define stabilizer codes, we first need to define the $n$-qubit Pauli group. To do this, consider the Pauli gates $X$, $Y$ and $Z$, defined respectively by the Pauli matrices $\sigma_x$, $\sigma_y$ and $\sigma_z$:
\begin{equation}
		X=\sigma_x=\begin{bmatrix}
			0&1\\
			1&0
		\end{bmatrix}\text{, }Y=\sigma_y=\begin{bmatrix}
			0&-i\\
			i&0
		\end{bmatrix}\text{ and }Z=\sigma_z=\begin{bmatrix}
			1&0\\
			0&-1
		\end{bmatrix}.
	\end{equation}
The $n$\textit{-qubit Pauli group} $G_n$ is given by
\begin{equation}
		G_n=\left\{c\bigotimes_{i=1}^nA_i:c\in\{1,-1,i,-i\}\text{ and }A_i \in \{I,X,Y,Z\}\text{ for }1\leq i\leq n\right\}.
\end{equation}

Let $S$ be a subgroup of $G_n$ and let $V_S$ be the vector space consisting of all $n$-qubit states that are fixed by all elements of $S$. Then $V_S$ is said to be the \emph{vector space stabilized by} $S$, and $S$ is called the \emph{stabilizer} of $V_S$.

We say that the elements $g_1 , \ldots , g_r \in G_n$ are \emph{independent} if 
$$\langle g_1 , \ldots, g_{i-1} , g_{i+1}, \ldots, g_r\rangle \neq \langle g_1 , \ldots , g_r \rangle.$$ In other words, by deleting any generator $g_i$, the new subgroup generated becomes smaller.

The following proposition will pave the way for the definition of a stabilizer code.

\begin{proposition}\label{stabilizercode}
    \textup{\textbf{(Nielsen and Chuang~\cite[Page 458]{Nielsen_Chuang_2010})}} Let $S\subseteq G_n$ be a subgroup generated by $n-k$ independent and commuting elements from $G_n$ such that $-I \notin S$. Then $V_S$ is a $2^{k}$-dimensional vector space.
\end{proposition}

An $[[n, k, d]]$ \textit{stabilizer code} $C(S)$ is the vector space stabilized by a subgroup $S \subseteq G_n$ such that $-I\notin S$ and $S$ has $n-k$ independent and commuting generators. The \textit{normalizer} of a stabilizer $S$, denoted by $N(S)$, is the set of all elements $E$ of $G_n$ such that $EgE^\dagger \in S$ for all $g \in S$. The \textit{weight} of an element $E \in G_n$ is the number of terms different from identity in the tensor product that forms it. The \textit{minimum distance} $d$ of $C(S)$ is the minimum weight of an element of $N(S)-S$. In this paper, we take into consideration only binary stabilizer codes, but nonbinary stabilizer codes were already characterized in~\cite{non_binary_finite} in a generalization of additive codes over $\mathbb{F}_4$~\cite{quantum_gf4}.

Next, we will define the CSS codes, a particular class of stabilizer codes.

Let $C_1$ and $C_2$ be two linear binary codes $[n,k_1]$ and $[n,k_2]$, respectively, such that $C_2 \subset C_1$ and  $C_1$ and $C_2^\perp$ both correct $t$ errors. Then the \textit{Calderbank-Shor-Steane (CSS)} code~\cite{calderbank_shor_css,steane_css} of $C_1$ over $C_2$, denoted by $\text{CSS}(C_1,C_2)$, is the stabilizer code with check matrix
\begin{equation}
		\begin{bmatrix}
			\begin{array}{c|c}
				H(C_2^\perp) & 0 \\
				0 & H(C_1)
			\end{array}
		\end{bmatrix}.
	\end{equation}

\section{The Results}
\label{sec:results}

In this section, we present the results of the paper. The main result is the construction of new linear codes with maximum minimum distance (see Theorem~\ref{theorem:new_codes}). 
In Theorem~\ref{theorem:graph_fields}, we generalize the construction of linear codes based on the incidence matrix of a given graph to all finite fields. Another contribution is to construct CSS quantum codes derived from cycle graphs by applying Theorem~\ref{theo:matroid-code-subset}.

\subsection{CSS Codes from Cycle Graphs}

\begin{lemma}
    \label{lemma:spanning-set}
    Let $V$ be a vector space over ${\mathbb F}_2$ and $S \subseteq V$ the set of all non-zero vectors in $V$ with minimal support. Then, $S$ is a spanning set of $V$.
\end{lemma}

\begin{proof}
    Let $v\in V$ a vector and $\supp(v)$ the support of $v$. We will show that $v$ can be obtained through the sum of vectors in $S$. If $v\in S$, there is nothing to prove. If $v \notin S$, there exists $w_1 \in S$ such that $\supp(w_1) \subsetneq \supp(v)$. Let $w_2 \in V$ the vector whose support $\supp(w_2)$ consists of all non-zero coordinates of $v-w_1$, i.e., $v= w_1 + w_2$. If $w_2 \in S$, the proof is concluded; otherwise, there exists $w_3 \in S$ such that $\supp(w_3) \subsetneq\supp(w_2)$. Let $w_4 \in V$ be the vector whose support consists of the non-zero coordinates of $w_2 - w_3$. Then, $v= w_1 + w_3 + w_4$. We know the process is finite, hence there is a finite number of vectors $w_i \in S$, $i=1, 2, \ldots ,t$ such that $v=w_1 + w_2 + \ldots + w_t$. This concludes the proof.
\end{proof}

\begin{lemma}
    \label{lemma:support_circuit}
    Let $A$ be an $m\times n$ matrix, $M$ be the vector matroid of $A$ and $S$ the set of all non-zero vectors in $\ker(A)$ with minimal support. Let $v=[v_0,\dots,v_{n-1}]$ be a vector. Then $v \in S$ if and only if $\supp(v) \in \mathcal{C}(M)$.
\end{lemma}
\begin{proof}
    Let $v$ be a member of $S$. Since $v \in \ker(A)$, we have
    \begin{equation}
        \label{eq:kernel_support1}
			Av=\sum_{i=0}^{n-1}a_{i} v_i=\sum_{i \in \supp(v)}a_i=0,
    \end{equation}
    where $a_{i}$ is the $i$th column of $A$, for $i=0,1,\dots,n-1$. The last equality of Equation~\ref{eq:kernel_support1} implies that the columns of $A$ labeled by $\supp(v)$ are linearly dependent, thus $\supp(v) \in \mathcal{D}(M)$, where $\mathcal{D}(M)$ is the collection of dependent sets of $M$. Suppose that $\supp(v) \notin \mathcal{C}(M)$. In this case, there would be $C \subsetneq \supp(v)$ such that $C \in \mathcal{C}(M)$ and we would have
    \begin{equation}
			0=\sum_{i \in C}a_i =\sum_{i=0}^{n-1}a_iw_i=Aw
	\end{equation}
    where $w$ is the binary vector with support $C$, a contradiction, since $\supp(v)$ is minimal. Thus, $\supp(v) \in \mathcal{C}(M)$.

    Conversely, let $C$ be a circuit of $M$. By definition, we have
    \begin{equation}
			\label{eq:kernel_support2}
			0 = \sum_{i \in C}a_i = \sum_{i=0}^{n-1}a_iv_i = Av
	\end{equation}
    where $v$ is the binary vector with support $\supp(v)=C$. By Equation~\ref{eq:kernel_support2},  we see that $v \in \ker(A)$. Suppose that $v \notin S$. In this case, there would be a vector $w \in \ker(A)$ such that $\supp(w) \subsetneq \supp(v)$ and we would have
    \begin{equation}
			Aw=\sum_{i=0}^{n-1}a_iw_i=\sum_{i \in \supp(w)}a_i=0,
	\end{equation}
    where the last equality would imply that $\supp(w) \subsetneq \supp(v) = C$ is a dependent set of $M$, a contradiction, since $C$ is a minimal dependent set. Thus, $v \in S$.
\end{proof}

\begin{theorem}
    \label{theo:matroid-code-subset}
    Let $M_1$ and $M_2$ be $\mathbb{F}_2$-representable matroids with the same ground set $E$ and $H_1$ and $H_2$ representations of $M_1$ and $M_2$, respectively. Let $C_1=\ker(H_1)$ and $C_2=\ker(H_2)$ be two linear codes over $\mathbb{F}_2$. Then, if every circuit of $M_1$ is equal to a disjoint union of circuits of $M_2$, $C_1$ is contained in $C_2$.
\end{theorem}

\begin{proof}
    Let $S_1$ and $S_2$ be the sets of non-zero vectors of the null spaces of $H_1$ and $H_2$, respectively, with minimal support. Let $v$ be a vector in $S_1$. From hypothesis $\supp(v)=\bigcup_{i=0}^{k-1}C_i$, where $C_i \in \mathcal{C}(M_2)$ for all $i$, and $C_i\cap C_j=\emptyset$ for all $0\leq i,j \leq k-1$. Since by Lemma~\ref{lemma:support_circuit} $C_i=\supp(w_i)$ for some binary vector $w_i \in S_2$ and the support of $v$ is the union of the supports of all $w_i$ and there is no intersection between them, we can write $\sum_{i=0}^{k-1}w_i$. Therefore, since all $w_i$ belong to $S_2 \subsetneq C_2$ and $C_2$ is closed under vector addition, $v \in C_2$. Thus $S_1 \subsetneq C_2$. By Lemma~\ref{lemma:spanning-set}, $S_1$ is the spanning set of $\ker(H_1)=C_1$ and, since $S_1 \subseteq C_2$, $C_1 \subseteq C_2$.
\end{proof}

Suppose we have a cycle graph $G_1=(V_1,E_1)$ with $n$ edges. Let $G_2=(V_2,E_2)$ be a disconnected graph with $n$ edges whose components are cycle graphs. Now suppose that we construct cycle matroids $M_1$ and $M_2$ from $G_1$ and $G_2$, respectively. Since $|E_1|=|E_2|$, we can make $E(M_1)=E(M_2)=E$ by mapping the edges in $E_1$ and $E_2$ to the same set of elements $E$ and use $E$ as the ground set of $M_1$ and $M_2$. As there is a one-to-one correspondence between cycles of a graph and circuits of the corresponding cycle matroid, clearly all circuits of $M_2$ are contained in the only circuit of $M_1$. Let $C_1$ and $C_2$ be the binary linear codes corresponding to $M_1$ and $M_2$, respectively, that is, the codes that are the null spaces of a representation of the matroids. Then, by Theorem~\ref{theo:matroid-code-subset}, $C_2$ contains $C_1$. Thus, we can use cycle graphs to generate linear codes that can be used to construct CSS codes.

As an example, consider the graphs $G_1$ and $G_2$ in Figure~\ref{fig:grafos}. From them we can construct two matroids $M_1$ and $M_2$ with ground set $E=\{0,1,2,3,4,5\}$ and set of circuits $\mathcal{C}_1=\{\{0,1,2,3,4,5\}\}$ and $\mathcal{C}_2=\{\{0,1,2\},\{3,4,5\}\}$, respectively. A possible representation of $M_1$ is the matrix
\begin{equation}
		\begin{bmatrix}
			0&1&0&0&0&1\\
			0&0&0&0&1&1\\
			0&0&1&0&0&1\\
			0&0&0&1&0&1\\
			1&0&0&0&0&1
		\end{bmatrix},
	\end{equation}
which is also a parity-check matrix of the code $C_1=\{000000,111111\}$, with parameters $[6,1,6]$. The matrix
\begin{equation}
		\begin{bmatrix}
			1&0&0&1&0&0\\
			0&0&1&0&0&1\\
			0&1&0&1&0&0\\
			0&0&0&0&1&1
		\end{bmatrix}
	\end{equation}
is a representation of $M_2$ and a parity-check matrix of the code \linebreak\mbox{$C_2=\{000000,110100,001011,111111\}$} with parameters $[6,2,3]$. Therefore, the code $CSS(C_2,C_1)$ with parameters $[[6,1,2]]$, $d_X=3$ and $d_Z=2$, has check matrix
\begin{equation}
		\begin{bmatrix}
			1&1&1&1&1&1&0&0&0&0&0&0\\
			0&0&0&0&0&0&1&0&0&1&0&0\\
			0&0&0&0&0&0&0&0&1&0&0&1\\
			0&0&0&0&0&0&0&1&0&1&0&0\\
			0&0&0&0&0&0&0&0&0&0&1&1\\
		\end{bmatrix}.
	\end{equation}

	\begin{figure}[htbp] 
		\centering
		\begin{subfigure}[b]{0.45\textwidth}
			\centering
			\includegraphics[width=\textwidth]{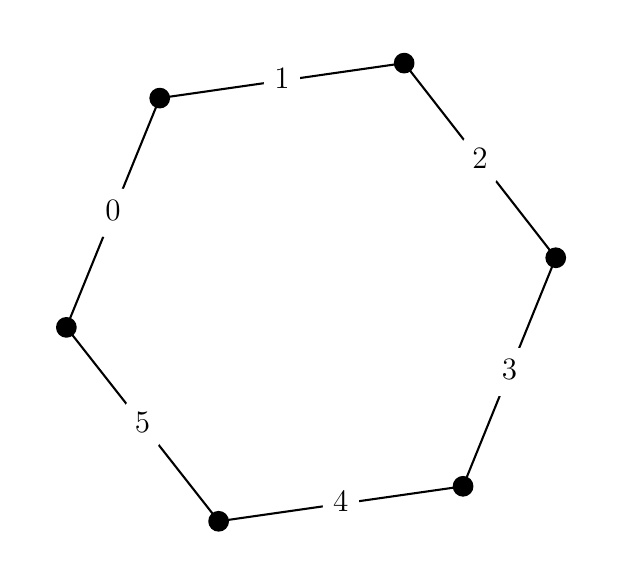}
			\caption{$G_1$}
			\label{fig:subfig1}
		\end{subfigure}
		\hfill
		\begin{subfigure}[b]{0.45\textwidth}
			\centering
			\includegraphics[width=\textwidth]{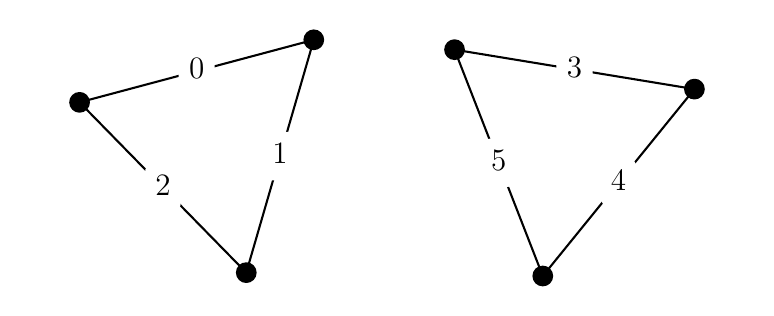}
			\caption{$G_2$}
			\label{fig:subfig2}
		\end{subfigure}
		
		\caption{6-edge graphs}
		\label{fig:grafos}
	\end{figure}

\subsection{Linear Codes from Book Graphs}

In this section, we will show how linear codes can be built from the $2$-sum of two isomorphic book graphs $B_{n,c}$. First, we prove a lemma that relates the properties of a matroid  $M$ to those of the null space of a matrix representation of $M$, which, by definition, is a linear code.

\begin{lemma}
\label{lemma:matroide-code-params}
    Let $M=(E,\mathcal{I})$ be a matroid with rank $r$ and girth $g$, $H$ a representation of $M$ and $C=\ker(H)$ the null space of $H$. Then $C$ is a $[|E|, |E|-r,g]$ linear code.
\end{lemma}

\begin{proof}
    Clearly the length of the vectors in the null space of $H$ is equal to its number of columns, i.e., the length of $C$ is equal to $|E(M)|$. If $n$ and $k$ are, respectively, the length and the dimension of $C$, the rank of $H$ is $n-k$. In addition, since $H$ is a representation of $M$, the rank of $H$ is also equal to $r_M$. Thus, $k=n-r_M=|E(M)|-r_M$. By definition, the minimum distance $d$ of $C$ is the minimum weight of a non-zero vector in $\ker(H)$. Let $v \in \ker(H)$ be a vector such that $\text{w}(v)=d$. The support of $v$ is clearly minimal and, by Lemma~\ref{lemma:support_circuit}, corresponds to a circuit of $M$. As the support of $v$ is minimal among the non-zero vectors from $\ker(H)$, we have $d=\text{w}(v)=|\supp(v)|=g(M)$. This concludes the proof.
    
\end{proof}

\begin{theorem}
\label{theorem:graph_fields}
    Let $G$ be a graph and $D$ an arbitrary orientation of $G$. Let $A_D$ an incidence matrix of $D$ in which the columns corresponding to loops have only zeroes. Then the following statements about $A_D$ are true:
    \begin{enumerate}[label=(\roman*)]
        \item $A_D$ can be interpreted as a matrix over the field $\mathbb{F}_q$ for any $q=p^m$, where $p\geq 2$ is a prime number, having as entries the multiplicative and additive identities and $p-1$.
        \item The null space $C$ of $A_D$ is an $[|E|, |E|-|V|+\omega,g]$ linear code over $\mathbb{F}_q$, where $E$ is the set of edges of $G$, $V$ is the set of vertices of $G$, $\omega$ is the number of connected components of $G$ and $g$ is the girth of $G$.
        \item If $H$ can be obtained from $A_D$ through row reduction, then its entries are the same from $A_D$, that is, each entry can be only $-1$, $0$ or $1$.
        \item The null space of a matrix over $\mathbb{F}_{q'}$, with $q'={(p')}^{m'}$ and $p'\neq p$, obtained by replacing $p-1$ by $p'-1$ is a linear code with the same parameters as $C$.
    \end{enumerate}
\end{theorem}

\begin{proof}
\begin{enumerate}[label=(\roman*)]
    \item Since $A_D$ has only $0$, $1$, and $-1$ as entries, we can interpret it as a matrix over the field $\mathbb{F}_{q}$ in which $0$ is the additive identity, $1$ is the multiplicative identity and $-1$ is the additive inverse of $1$. Thus, $-1$ corresponds to $p-1$. This proves $(i)$.

   \item By Proposition~\ref{prop:rep-graphic}, $A_D$ represents the matroid $M(G)$ over $\mathbb{F}_q$ and by Lemma~\ref{prop:rank-cyclic}, $r_{M(G)}=|V|-\omega$. Thus, by Lemma~\ref{lemma:matroide-code-params}, $C=\ker(A_D)$ is a $[|E|, |E|-|V|+\omega,g]$ code.

    \item Row reduction uses the row swap, multiplication of a row by a scalar and the sum of a row to the multiple of other. To prove $(iii)$, we need to show that each one of this operations adds no different entry or is not necessary. Row swap clearly does not change the matrix entries. The multiplication by a scalar is necessary in row reduction to generate the additive inverse of the row to be made 0, but in a column, in a addition to zeros, there are only a $1$ and a $p-1$, that are additive inverses of each other in a field with characteristic $p$. thus the multiplication by scalar is not necessary. Finally, we have the sum of a row to the multiple of other, which we will consider simply as the sum of two row, since we do not need the multiple. In the sum of two rows, in each column we will have the sum of a $0$ to $1$ or $p-1$, that does not generate a new entry, or the sum of $1$ to $p-1$, that results in 0. Thus, row reduction adds no new entry.

    \item Let $C=\ker(H)$ be a linear code over a field with characteristic $p$, where $H$ is obtained from $A_D$ through row reduction. Let $C'=\ker(H')$ be a linear code where $H'$ is the reduced matrix of $A_D$ over a field with characteristic $p'$. Notice that all possible sums in the row reduction of $A_D$, $0+1$, $0+(-1)$ and $1+(-1)$, seen in the proof of Item $(iii)$, will result in one of the original entries and does not depend on the characteristic of the field. Thus, if we write $p-1$ and $p'-1$ as $-1$ in $H$ and $H'$, respectively, we have $H\approx H'$, that is, the entries in the same row and column in $H$ and $H'$ are written the same way despite representing different values in their respective fields. Therefore, replacing $p-1$ by $p'-1$ in $H$, we get $H'$. Since $C'=\ker(H')=\ker(A_D)$ ($A_D$ over a field with characteristic $p'$) and $C=\ker(H)=\ker(A_D)$ ($A_D$ over a field with characteristic $p$), $C$ and $C'$ have the same parameters by Item $(ii)$.
\end{enumerate}
   
\end{proof}

In the following proposition, we prove some properties of book graphs. They will be useful to determine the parameters of the codes obtained by the 2-sum of book graphs.

\begin{proposition}
\label{prop:book-graph-props}
    A book graph $B_{n,c}$ has the following properties:

    \begin{enumerate}
        \item $|V(B_{n,c})|=c(n-2)+2$.
        \item $|E(B_{n,c})|=c(n-1)+1$.
        \item $g(B_{n,c})=n$.
    \end{enumerate}
\end{proposition}

\begin{proof}
    \begin{enumerate}
        \item All cycles sharing exactly one edge means they share exactly two vertices, thus each one has $n-2$ exclusive vertices. Since $B_{n,c}$ has $c$ cycles, there are $c(n-2)$ of these exclusive vertices. Thus, adding the two shared vertices, the graph has a total of $c(n-2)+2$ vertices.
        \item Since all cycles shared exactly one edge, each one has $n-1$ exclusive edges. Since $B_{n,c}$ has $c$ cycles, there are $c(n-1)$ of these exclusive edges. Thus, adding the shared edge, the graph has a total of $c(n-1)+1$ edges.
        \item We will prove this by induction on the number of cycles. Consider the base case with just one cycle, i.e., $B_{n,1}$. Since it is a $n$-edge cycle graph, its girth is clearly $n$. Next, assume that $g(B_{n,c})=n$ holds. We must show that $g(B_{n,c+1})=n$. Let $e=(v_1,v_2)$ be the edge shared by the $c$ cycles in $B_{n,c}$. To obtain $g(B_{n,c+1})=n$ from $B_{n,c}$, we can add a $(n-1)$-edge path $P=u_1,e_1,u_2,e_2,\dots ,e_{n-1},u_n$ and merge $v_1$ with $u_1$ and $v_2$ with $u_n$. Notice that $v_1,e_1,u_2,e_2,\dots,e_{n-1}v_2,e,v_1$ is a $n$-edge cycle. Now, suppose the addition of $P$ creates a cycle $C=x_1,f_1,x_2,f_2,\dots,x_m,f_m,x_1$ with $m < n$. Since $C$ does not exist in $B_{n,c}$, at least one edge in $C$ must be in $P$. Suppose edge $f_i$, for $1 \leq i \leq m$, belongs to $P$. Suppose we walk along $C$ from $f_i$ in both directions. Since $P$ is a path, it cannot contain $C$, so our walk should pass by an edge from $B_{n,c}$ before returning to $f_i$. But the way we constructed $B_{n,c+1}$, the only connection between $B_{n,c}$ and $P$ are through vertices $v_1$ and $v_2$, so our walk will certainly pass by $v_1$ in one direction, after $i-1$ edges and by $v_2$, after $n-1-i$ edges. With this, we traversed $n-1$ edges without covering the whole cycles, implying that $m>n-1$, a contradiction.

    \end{enumerate}
\end{proof}

\begin{theorem}
\label{theorem:new_codes}
    There exist linear codes with the greatest possible minimum distance for its corresponding length and dimension, which are constructed by means of the null space of a representation of a $2$-sum of the cycle matroids of two isomorphic book graphs.
\end{theorem}

\begin{proof}
    Consider the cycle matroid $M(B_{n,c})$. By Proposition~\ref{prop:book-graph-props}, it is easy to see that the number of elements of $M(B_{n,c})$ is $c(n-1)+1$, its rank is $c(n-2)+1$ and its girth is $n$. Now consider the $2$-sum $N=M(B_{n,c})\oplus_{2}M(B_{n,c})$ where $p$ is the edge shared by the cycles of $B_{n,c}$. By Proposition~\ref{prop:2srank}, we know that $r_N=2c(n-2)+1$. Let $C(N)$ be a linear code such that $N$ is the vector matroid of its parity-check matrix. Then, $C(N)$ has length $2c(n-1)$, dimension $2c-1$ and distance $2n-2$.

    Consider the collection of circuits  $\mathcal{C}$ of the 2-sum of two matroids $M_1$ and $M_2$, given by
    \begin{equation}
        \mathcal{C}=\mathcal{C}(M_1\setminus p) \cup \mathcal{C}(M_2\setminus p) \cup \{C_1 \cup C_2:p\in C_1 \in \mathcal{C}(M_1),p\in C_2 \in \mathcal{C}(M_2)\}.
    \end{equation}
    From this, we can see that the girth of the resulting matroid will be greater than $g(M_1)$ and $g(M_2)$ if the element $p$ is not in the smallest circuits of $M_1$ and $M_2$, since, in this case, these circuits will be preserved. To prevent this and increase the girth through the 2-sum, we used book graphs, since their all smallest cycles share a single edge and using the corresponding matroid element as $p$, eliminates the smallest cycles in the 2-sum, allowing the girth to be increased.

   In the case of the graphs $B_{3,2}$, $B_{3,3}$, $B_{4,2}$, $B_{5,2}$, the binary codes constructed have the greatest possible minimum distance for their length and dimension. In addition, with the exception of the one based on  $B_{4,2}$, all of them are self-orthogonal, which means they can also be used to construct CSS codes. \end{proof}

In Table~\ref{tab:cod-book-graph}, we exhibit new codes that we constructed in this paper which have maximum minimum distance, namely, $C_2$, $C_3$, $C_4$. Moreover, the code $C_1$ is not new but it was constructed by a distinct method than the available in the literature. In Figure~\ref{fig:generator-matrices}, we exhibit the generator matrices of these codes.

    \begin{table}
    \centering
    \begin{tabular}{c|c|c}
        Code & Graph & Parameters\\
        \hline
        $C_1$ & $B_{3,2}$ & $[8,3,4]$\\
        $C_2$ & $B_{3,3}$ & $[12,5,4]$\\
        $C_3$ & $B_{4,2}$ & $[12,3,6]$\\
        $C_4$ & $B_{5,2}$ & $[16,3,8]$\\
        \hline
    \end{tabular}
    \caption{Linear codes based on the 2-sum of book graphs.}
    \label{tab:cod-book-graph}
\end{table}

 \begin{figure}[htbp]
    \centering

    \begin{subfigure}{0.9\columnwidth}
        \centering
        \begin{equation*}
            \begin{bmatrix}
                1 & 0 & 1 & 0 & 0 & 1 & 0 & 1 \\
                0 & 1 & 0 & 1 & 0 & 1 & 0 & 1 \\
                0 & 0 & 0 & 0 & 1 & 1 & 1 & 1 \\
            \end{bmatrix}
        \end{equation*}
        \caption{Code $C_1$}
        \label{fig:code_c1}
    \end{subfigure}

    \vspace{0.5cm}

    \begin{subfigure}{0.9\columnwidth}
        \centering
        \begin{equation*}
            \begin{bmatrix}
                1 & 0 & 0 & 0 & 1 & 0 & 0 & 0 & 1 & 0 & 0 & 1 \\
                0 & 1 & 0 & 1 & 0 & 0 & 0 & 0 & 1 & 0 & 0 & 1 \\
                0 & 0 & 1 & 0 & 0 & 1 & 0 & 0 & 1 & 0 & 0 & 1 \\
                0 & 0 & 0 & 0 & 0 & 0 & 1 & 0 & 1 & 0 & 1 & 1 \\
                0 & 0 & 0 & 0 & 0 & 0 & 0 & 1 & 1 & 1 & 0 & 1 \\
            \end{bmatrix}       
        \end{equation*}
        \caption{Code $C_2$}
        \label{fig:code_c2}
    \end{subfigure}

    \vspace{0.5cm}

    \begin{subfigure}{0.9\columnwidth}
        \centering
        \begin{equation*}
            \begin{bmatrix}
                1 & 1 & 0 & 1 & 0 & 0 & 0 & 0 & 1 & 0 & 1 & 1 \\
                0 & 0 & 1 & 0 & 1 & 1 & 0 & 0 & 1 & 0 & 1 & 1 \\
                0 & 0 & 0 & 0 & 0 & 0 & 1 & 1 & 1 & 1 & 1 & 1 \\
            \end{bmatrix}
        \end{equation*}
        \caption{Code $C_3$}
        \label{fig:code_c3}
    \end{subfigure}

    \vspace{0.5cm}

    \begin{subfigure}{0.9\columnwidth}
        \centering
        \begin{equation*}
            \begin{bmatrix}
                1 & 0 & 0 & 0 & 1 & 0 & 1 & 1 & 0 & 1 & 1 & 1 & 0 & 1 & 0 & 0 \\
                0 & 1 & 1 & 1 & 0 & 1 & 0 & 0 & 0 & 1 & 1 & 1 & 0 & 1 & 0 & 0 \\
                0 & 0 & 0 & 0 & 0 & 0 & 0 & 0 & 1 & 1 & 1 & 1 & 1 & 1 & 1 & 1 \\
            \end{bmatrix}
        \end{equation*}
        \caption{Code $C_4$}
        \label{fig:code_c4}
    \end{subfigure}

    \caption{Generator matrices of codes $C_1$ to $C_4$}
    \label{fig:generator-matrices}
\end{figure}

\section{Final Remarks}\label{fr}

In this paper, we have shown a new method to construct linear codes by means of the $2$-sum of two isomorphic book graphs and, applying such a method, we found four binary codes with the greatest minimum distance possible for their length and dimension. We have also shown how to generalize codes obtained from incidence matrices to finite fields with any characteristic, keeping the same length, dimension and minimum distance, which gives us great flexibility when generating codes from graphs. Moreover, we have shown how to determine in which cases two linear codes corresponding to two matroids can be utilized to construct a quantum CSS code. This gives us a way to construct new families of quantum codes using matroid theory. In addition, we have shown a way to construct CSS codes from cycle graphs using this result.

\FloatBarrier

\section*{Acknowledgements}

The authors thank R. Portugal for helpful discussions.
A.S.B.~acknowledges the financial support from CAPES Process No. 88887.670714/2022-00.
F.L.M.~acknowledges financial support from CNPq Grant
No. 06049/2025-2.
G.G.LaG.~acknowledges financial support from CNPq Grant
No. 302984/2022-4.
The authors acknowledge the financial support of the National Institute of Science and Technology for Applied
Quantum Computing through CNPq process No. 408884/2024-0.

\paragraph{Conflict of interest.} The authors declare no conflicts of interest.

\paragraph{Data availability.} No new data were created or analyzed in this study. This is purely a theoretical work, and all mathematical derivations and conclusions are explicitly detailed within the manuscript.

\bibliographystyle{abbrv}
\bibliography{referencias}
\end{document}